\documentclass[lettersize,journal]{IEEEtran}

\usepackage{amsmath,amssymb,amsfonts}
\usepackage{graphicx}
\usepackage{subfig}
\usepackage{textcomp}
\usepackage{xcolor}
\usepackage{cite}
\usepackage{hyperref}
\usepackage{amsthm}

\newtheorem{theorem}{Theorem}[section]
\newtheorem{proposition}[theorem]{Proposition}
\newtheorem{lemma}[theorem]{Lemma}
\newtheorem{corollary}[theorem]{Corollary}
\newtheorem{definition}[theorem]{Definition}

\begin{document}

\title{Spectral Graph Analysis for Predicting QoE Fairness Sensitivity in Wireless Communication Networks}

% \author{
%     % Author names and affiliations go here
%     Xinke Jian,
%     Zhiyuan Ren,~\IEEEmembership{Member,~IEEE,}
%     Wenchi Cheng,~\IEEEmembership{Senior Member,~IEEE,}% <-this % stops a space
% \thanks{Zhiyuan Ren, Xinke Jian and Wenchi Cheng are with the School of Telecommunications Engineering, Xidian University, Xi'an 710071, China.}% <- 2. 第一个单位
% \thanks{Corresponding author: Zhiyuan Ren (zyren@xidian.edu.cn).}%
% }

\author{
    \IEEEauthorblockN{
    Xinke Jian\IEEEauthorrefmark{1}, 
    Zhiyuan Ren\IEEEauthorrefmark{1}, 
    Wenchi Cheng\IEEEauthorrefmark{1}}
    
    \IEEEauthorblockA{\IEEEauthorrefmark{1}State Key Laboratory of Integrated Services Networks, Xidian University, Xi'an, China}
    
    \IEEEauthorblockA{E-mail:\{xkjian@stu.xidian.edu.cn, zyren@xidian.edu.cn, wccheng@xidian.edu.cn\}}

\thanks{This work was supported by the National Key Research and Development Program of China (No. 2024YFE0200302).}%
}
% \thanks{This work was supported by the National Key Research and Development Program of China (No. 2024YFE0200302).}%

% % The paper headers
% \markboth{IEEE/ACM TRANSACTIONS ON NETWORKING, VOL. XX, NO. X, MONTH 2025}%
% {Shell \MakeLowercase{\textit{et al.}}: A Spectral Bottleneck for Fairness}

\maketitle

\begin{abstract}
The evaluation of Quality of Experience (QoE) fairness depends not only on its current state but, more critically, on its sensitivity to changes in Service Level Agreement (SLA) parameters. 
However, the academic community has long lacked a predictive method connecting underlying topology to high-level service fairness. To bridge this gap, this paper analyzes a QoE imbalance index ($I$) through the lens of spectral graph theory.
Our core contribution is the proof of a novel exponential spectral upper bound. 
This bound reveals that the improvement of QoE fairness exhibits an exponential decay behavior only above a performance threshold determined jointly by network size and connectivity. 
Its core decay rate is dominated by the weaker of two factors: the SLA stringency ($a$) and the network's spectral gap ($c\lambda_2$). 
The upper bound unifies the service protocol and the topological bottleneck within a single performance bound formula for the first time.
This theoretical relationship also reveals a clear bottleneck effect, where the system's fairness ceiling is determined by the weaker link between service parameters and network structure. 
This finding provides a bottleneck-driven principle for resource optimization in network design and enables goal-driven reverse engineering. Extensive numerical experiments on various random graph models and real-world network topologies robustly validate the correctness and universality of our analytical framework.
\end{abstract}


\section{Introduction}
\label{sec:introduction}

As Internet services diversify and become increasingly complex, Quality of Experience (QoE) is increasingly governed by stringent Service Level Agreement (SLA) parameters. 
Unlike single-service performance optimization, multiplex numerous services share critical resources such as transmission links or core nodes in practical networks, 
which inherently induces competition. Under tightened SLA constraints and limited resources, scheduling and resource allocation could amplify congestion and concentrate resources, 
leading to highly heterogeneous QoE outcomes and hence QoE fairness degradation, as shown in Fig. 1 \cite{17}. QoE fairness concerns whether QoE outcomes are relatively balanced among services/users under shared resource conditions. 
In particular, it focuses not only on the system-wide average QoE, but also on the dispersion of experiences and the guarantees of tail-end services.

% \begin{figure*}[!t]
% \centering
% \includegraphics[width=\textwidth]{lamda2 VS I_v2/SLA3.png}
% \caption{Network traffic distribution graph. Subgraph (a) shows the network traffic distribution under loose SLA parameters, while subgraph (b) shows the network traffic distribution under strict SLA parameters.}
% \label{fig:SLA}
% \end{figure*}

\begin{figure*}[!t]
  \centering
  \subfloat[Loose SLA ($h_0=5$ hops)]{
    \includegraphics[width=0.48\textwidth]{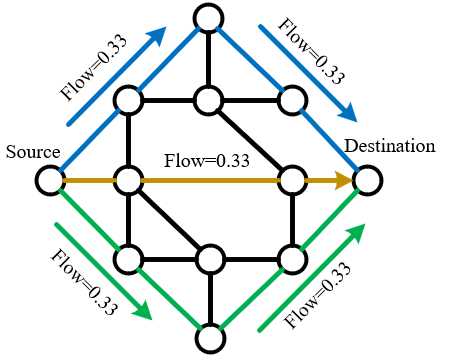}
    \label{fig:sla_loose}
  }\hfill
  \subfloat[Tight SLA ($h_0=3$ hops)]{
    \includegraphics[width=0.48\textwidth]{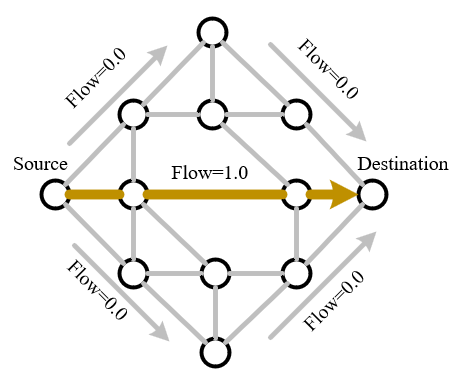}
    \label{fig:sla_tight}
  }
  \caption{Network traffic distribution graph. Subgraph (a) shows the traffic distribution under loose SLA parameters, while subgraph (b) shows the distribution under strict SLA parameters.}
  \label{fig:SLA}
\end{figure*}

Once QoE fairness becomes imbalanced in the network, it may lead to multifaceted consequences. 
On the user side, it manifests as unstable experience and reduced satisfaction. 
Tail-end services may cause service loss or failures of mission-critical tasks. 
On the network side, such imbalance is often associated with the coexistence of localized congestion and path polarization, which decrease resource utilization efficiency and weaken robustness under traffic surges or fault events. 
Therefore, studying QoE fairness in terms of measurement, mechanism analysis, and adaptive control is crucial for enhancing service predictability and SLA enforceability, and for supporting deterministic experience guarantees in next-generation networks.

Although substantial progress has been made in network resource allocation, traffic engineering, and adaptive optimization enabled by machine learning and neural networks, 
most existing studies primarily focus on improving network metrics under pre-specified objectives, which are largely defined at the QoS level. In contrast, a unified, interpretable, 
and predictive analytical framework describing the variation behavior and performance limits of QoE fairness under changing network conditions is still lacking.

To fill this theoretical void, this paper employs an entropy-based QoE imbalance index $I$, derived from information-theoretic axioms, to measure the concentration of end-to-end QoE \cite{16}. 
Our central question is not merely to evaluate a static fairness score but to uncover the intrinsic relationship governing the sensitivity and performance boundaries of $I$ as the network topology and service parameters $(a,h_0)$ vary. 
Based on the above, from the perspective of spectral graph theory, this paper establishes an exponential boundary relationship between the QoE imbalance index $I$ and the network's spectral gap $\lambda_2$ \cite{3}.
Meanwhile, this relationship exhibits a clear bottleneck effect: the system's overall rate of fairness improvement is determined by the weaker of the stringency of service parameters and the level of network connectivity. 
This quantitative relationship provides a theoretical cornerstone for network engineering to shift from ``passive response" to ``proactive design". 
It not only makes it possible to evaluate the worst-case performance upper bound for a fixed topology but also gives rise to a series of design principles, and allows us to infer specific topological requirements to meet target performance goals.

The main contributions of this paper are as follows:

\begin{enumerate}
    \item \textbf{Exponential Spectral Upper Bound}: We prove an exponential upper bound on the QoE imbalance index $I$, with a topology-dependent threshold $r^*$. The bound further yields a bottleneck characterization of the decay rate, governed by $\min\{a, c\lambda_2\}$, and bottleneck-driven design principle for achieving target performance.

    \item \textbf{Structural Intervention Principle}: For the normalized Laplacian, we provide a structural intervention criterion guided by its Fiedler vector (Theorem \ref{thm:fiedler_intervention}) that can monotonically improve the network's spectral gap $\lambda_2$, offering a concrete method to realize the bottleneck-driven design principle.

    \item \textbf{Data-Driven Certificate}: We propose a performance upper bound (Proposition \ref{prop:data_driven_cert}) that can be rapidly evaluated based on a single network-wide distance statistic, without requiring spectral computation.
\end{enumerate}

The remainder of this paper will elaborate on and experimentally validate these points. The structure is as follows: Section \ref{sec:related Work} analyzes the current research progress on QoE fairness; 
Section \ref{sec:framework} provides formal definitions; Section \ref{sec:main_results} presents the main theorem and key proof points; Section \ref{sec:engineering_implications} discusses its engineering implications; 
Section \ref{sec:experiments} conducts experimental validation on various networks; Section \ref{sec:conclusion} concludes the paper.

\section{Related Work}
\label{sec:related Work}

Existing studies on QoE fairness have made some progress in the relevant fields. Below, we review the literature from three representative perspectives: network resource allocation algorithms, robust traffic engineering, and machine-learning/reinforcement-learning-based detection and adaptation.

Firstly, in scenarios where multiple services compete for limited resources, resource allocation reshape the QoE distribution across services or users by altering quality of service (QoS) states such as throughput, delay, and bit error rate. 
Accordingly, in these studies, QoE fairness can be defined either on conventional QoS metrics, such as max-min fairness that prioritizes the worst service, or proportional fairness (PF) that balances efficiency and fairness \cite{4,5,6,7}, or on the distribution of user-perceived satisfaction through QoS-to-QoE mapping. 
Typical explicit modeling of QoE fairness appears in layered video multicast. For long term evolution (LTE) layered video multicast under heterogeneous channels \cite{21}, Han et al. constructed a joint decision framework and compared multiple allocation strategies, including throughput maximization, 
throughput-based PF, QoE maximization, and QoE-based proportional fairness, verifying that the QoE-based PF strategy can significantly improve user satisfaction with a slight average QoE loss. 
In video/multi-stream bandwidth allocation \cite{22}, Zhang introduced discounted scoring of user historical QoE and Jain's fairness to characterize multi-stream QoE fairness, and constructed a reward function by linearly combining average QoE, QoE changes, and fairness index. 
Then it learned distributed bandwidth weights that enhance the overall experience while promoting the balanced QoE distribution. Zhang et al. leveraged a graph neural network (GNN) to extract topology and trained user-level decision models in parallel, aiming to provide “fairness guarantees” for all users in an approximate Pareto-optimal sense \cite{23}. 
Wang et al. regarded QoE as the system utility and constructed a QoE-based maximization problem. It converted subcarrier allocation into weighted bipartite graph maximum matching and solved via the Kuhn-Munkres algorithm \cite{24}.

In addition, some studies have utilized traffic engineering(TE) to reallocate network resources via route selection and multi-path traffic splitting, so as to satisfy QoS requirements of services and further influence QoE \cite{8,9,10}. 
For reconfigurable data-center networks \cite{11}, Teh et al. proposed robust topology engineering approach that constructs an uncertain set from historical traffic and jointly optimizes topology and routing. It improved overall indicators such as maximum link utilization (MLU) under low-frequency reconfiguration conditions, 
and further correlated network indicators with application layer performance. This work mainly focused on enhancing network robustness and overall performance. For service-oriented software defined network (SDN) \cite{26}, M. Beshley et al. incorporated QoE levels into SLAs. 
The controller performed adaptive prioritization, multi-level routing and server selection to maintain user experience under different service requirements. Guo et al. proposed the SDN/OSPF traffic engineering (SOTE) framework in hybrid SDN environment, simultaneously optimizing the traffic-splliting ratios at SDN nodes and the open shortest path first (OSPF) weight settings. 
It modeled and solved this problem in a mixed-integer nonlinear programming form. The reported gains were primarily reflected by reduced MLU \cite{27}.

Overall, the above work provides classical fairness principles and actionable optimization goals for multi-service resource allocation. However, their focus mainly remained on how to allocate network resources under a given performance target. 
In most cases, the optimization is QoS-centric, with QoE and fairness often considered as higher-layer metrics and indirectly supported. As a result, direct constraints on the evolution of QoE fairness and exploration of the underlying mechanisms are still insufficient.

With the widespread application of machine learning and neural network models in network optimization and intelligent operations, a class of data-driven approaches has also emerged for QoE fairness. 
These methods employ deep models to learn the predictive relationship between QoE and network strategies, from high-dimensional network states, and then output resource allocation, scheduling or workload migration decisions to improve the distribution of user experiences\cite{12,13,14,15}. 
Representative works include deep reinforcement learning (DRL) with fairness-aware rewards\cite{28}, topology-aware GNN modeling with multi-objective training for multi-user fairness guarantees\cite{29}, and multi-agent reinforcement learning (RL) that encodes an explicit fairness-efficiency trade-off in the reward design\cite{30}. 
Despite their empirical effectiveness, the decision logic in these approaches is largely embedded in learned model parameters, rendering them essentially black-box predictors. Consequently, they usually lack interpretable mechanistic explanations and derivable structured conclusions, 
which limits to reveal the fundamental causes of QoE imbalance and to explain the sources of QoE fairness sensitivity to network changes.

\section{Theoretical Framework and Formal Definitions}
\label{sec:framework}

This section aims to establish a rigorous and clear formal framework for the subsequent theoretical analysis. 
We will define the core spectral properties describing the network's underlying topology and the QoE imbalance metric for measuring high-level service performance.

\subsection{Network Topology and Spectral Properties}
Let $G=(V,E)$ be a graph with $|V|=n$. In this paper, we are concerned with the set of unordered pairs of nodes, denoted by $K=\{\{u,v\}: u,v \in V, u \ne v\}$, with a total of $M = n(n-1)/2$ pairs. 
Logarithmic normalizations later in the paper are computed with respect to $\ln(M)$. Let the maximum and minimum degrees of the graph be $\Delta$ and $\delta$, respectively, and define the degree ratio as $\beta = \Delta/\delta$. 
Assumption A (Bounded Degree Ratio): We assume that $\beta$ is a bounded constant for the family of graphs under consideration.

To quantify the graph's topological structure, we employ standard tools from spectral graph theory. 
Let $W$ be the adjacency matrix of the graph, and $D$ be the diagonal degree matrix, where its diagonal elements $D_{ii}$ are the degrees of node $i$. 
We use the Normalized Laplacian operator $\mathcal{L}$ to characterize the graph's structure, defined as \cite{2}:
\begin{equation}
\mathcal{L} = I - D^{-1/2} W D^{-1/2}.
\end{equation}
where $I$ is the identity matrix. $\mathcal{L}$ is a positive semi-definite matrix, and its eigenvalues contain rich information about the graph's connectivity

The core topological metric in this paper is the second smallest eigenvalue $\lambda_2$ of $\mathcal{L}$, commonly known as the algebraic connectivity. 
$\lambda_2$ is a scalar that reflects the overall connectivity and expansion capability of the graph. A larger $\lambda_2$ value typically implies a better connected network structure without significant bottlenecks, 
where information or traffic can mix rapidly \cite{3} \cite{1}.

\subsection{Service-Layer QoE Imbalance}
In this paper, since we focus on the endogenous relationship between the underlying network topology and the service QoE performance,
we abstract the model as an idealized but practical motivated regime where path performance is dominated by the topology and links are approximately homogeneous.
In this regime, we model path cost by the hop count $h(u,v)$ of the shortest path for any pair of nodes $\{u,v\}$ and map this path-based cost to a QoE satisfaction.

To ground this mapping in established telecommunications engineering practices, 
we adopt the standardized method from the international telecommunication union telecommunication standardization sector (ITU-T) recommendations (such as G.1030), 
which uses a logistic curve (sigmoid function) to model the non-linear relationship between technical parameters and user satisfaction \cite{19}. 
We thus define the satisfaction weight $w_{uv}$ as follows:
\begin{equation}
w_{uv} = \frac{1}{1+\exp\!\big(a\,[h(u,v)-h_0]\big)},
\end{equation}
Here, the two SLA parameters have clear physical interpretations: $a>0$ is the evaluation stringency, which determines the steepness of the satisfaction curve around the threshold; 
$h_0>0$ is the performance threshold, representing the acceptable performance gate for the service.

The satisfaction weights of all node pairs form a distribution. By normalizing these weights, we implicitly assign equal a prior importance to each potential communicating pair.
So we obtain a pairwise probability distribution $\mathbf{p}$ that describes the satisfaction distribution across the entire network, where each element is:
\begin{equation}
p_{uv} = \frac{w_{uv}}{\sum_{\{x,y\} \in K} w_{xy}},
\end{equation}
Finally, we use the QoE imbalance index $I$ to quantify the deviation of this probability distribution $\mathbf{p}$ from the ideal state of fairness (i.e., the uniform distribution $\mathbf{u}$, 
where $u_{uv} = 1/M$). It is defined as the normalized Kullback-Leibler (KL) divergence \cite{20}:
\begin{equation}
I(a,h_0;G) := \frac{D_{KL}(\mathbf{p} || \mathbf{u})}{\ln(M)}.
\end{equation}
where $D_{KL}(\mathbf{p} || \mathbf{u}) = \sum_{\{u,v\} \in K} p_{uv} \ln(p_{uv}/u_{uv})$. This definition is equivalent to $1 - H(\mathbf{p})/\ln(M)$, where $H(\mathbf{p})$ is the Shannon entropy. We use the KL divergence form as it is more fundamental in information-theoretic proofs. The index ranges from $[0, 1]$, where $I=0$ corresponds to perfect fairness and $I=1$ to extreme unfairness.

To sum up, the QoE imbalance index $I$ interprets as a structural QoE fairness baseline characterizing how topology and SLA parameters jointly shape the distribution of satisfaction across the network,
rather than as a full model of all operator-specific details.

\section{Main Theoretical Results}
\label{sec:main_results}

This section presents the core theoretical contribution of this paper. 
Through a novel mathematical theorem, we connect the long-separated fields of network topology and high-level service fairness.

\subsection{The Exponential Spectral Upper Bound}
Our main theoretical result is the following Theorem \ref{thm:spectral_bound}, which provides a pre-conditioned, exponential worst-case guarantee for the service fairness of any given network topology, determined by its spectral properties.

\begin{theorem}[Spectral Upper Bound]
\label{thm:spectral_bound}
Under the assumption of a bounded degree ratio and using the normalized Laplacian, there exist constants $C, c > 0$ that depend only on the degree ratio. Furthermore, there exists a performance threshold $r_* \le C \frac{\ln n}{\lambda_2}$, determined jointly by the network size $n$ and the spectral gap $\lambda_2$.

When the performance threshold $h_0$ exceeds this gate (i.e., $h_0 \ge r_*$), let the effective threshold be $H = h_0 - r_*$. Then, the network's imbalance index $I$ satisfies the following upper bound:
\begin{equation}
I \;\le\; \frac{C\,(1+aH)}{\ln M}\,\exp\!\Big(-\min\{a,\,c\lambda_2\}\,H\Big).
\end{equation}
\end{theorem}

This theorem reveals several profound insights. 
First, the exponential improvement in fairness is not unconditional but is subject to an activation threshold $r_*$. 
Only when the service standard $h_0$ is high enough to cross this threshold, which is determined by the network's own structure, does the system enter a rapid ``fairness enhancement zone". 
Second, its exponential decay rate is governed by a bottleneck effect. We distill this core observation into the following proposition:

\begin{proposition}[Bottleneck Effect]
\label{prop:bottleneck_effect}
When $h_0 \ge r_*$, the exponential decay rate of the imbalance index $I$ is determined by the lesser of the service stringency $a$ and the scaled network spectral gap $c\lambda_2$. This means the potential for improving system fairness is limited by the bottleneck in either the service parameters or the network structure.
\end{proposition}

This quantitative relationship not only provides an upper bound for performance analysis but, more importantly, transforms Theorem \ref{thm:spectral_bound} from a passive analytical tool into a powerful, proactive engineering design guide.

\begin{corollary}[Convergence of Expander Graph Families]
\label{cor:expander_convergence}
If a family of graphs $\{G_n\}$ has a uniform lower bound on its spectral gap, $\lambda_2(G_n)\ge \mu_0>0$ (i.e., the network is sufficiently "robust"), then as long as the performance threshold $h_0$ grows at least logarithmically ($h_0 \ge C\ln n$), its imbalance index $I(a,h_0;G_n) \to 0$ as the network size $n\to\infty$.
\end{corollary}

Corollary \ref{cor:expander_convergence} paints an ideal picture, but in more practical engineering scenarios, we face the reverse design problem: achieving a target under given constraints. 
Theorem \ref{thm:spectral_bound} provides a design manual for two such directions. To simplify the design formulas, we can neglect the linear pre-factor $(1+aH)$, which has a smaller impact when $H$ is large.

\begin{corollary}[Reverse Design: Threshold Lower Bound]
\label{cor:reverse_design_h0}
Under the setting of Theorem \ref{thm:spectral_bound}, to ensure the imbalance index $I$ does not exceed a target value $I_{\text{tar}}$, a sufficient service performance threshold $h_0$ approximately needs to satisfy:
\begin{equation}
h_0 \ge r_* + \frac{1}{\min\{a,\,c\lambda_2\}}\left[\ln\!\frac{C}{I_{\text{tar}}}-\ln\!\big(\ln M\big)\right].
\end{equation}
\end{corollary}

\begin{corollary}[Reverse Design: Spectral Gap Lower Bound]
\label{cor:reverse_design_lambda2}
Under the setting of Theorem \ref{thm:spectral_bound}, given service parameters $(a,h_0)$ and a performance target $I_{\text{tar}}$, if $h_0>r_*$, a sufficient condition for the network spectral gap $\lambda_2$ is approximately:
\begin{equation}
\lambda_2 \ge \frac{1}{c(h_0-r_*)}\left[\ln\!\frac{C}{I_{\text{tar}}}-\ln\!\big(\ln M\big)\right].
\end{equation}
(This formula assumes the system is in a structure-limited state, i.e., $a$ is sufficiently large).
\end{corollary}

\subsection{Discussion on the Tightness of the Bound}
It is important to emphasize that the tightness of the upper bound in Theorem \ref{thm:spectral_bound} varies across different types of networks. The theoretical rate is relatively tighter on topologies with clear structural bottlenecks (such as a path graph $P_n$), where it aligns well with the upper limit of actual performance. In contrast, for highly connected topologies (like expander graphs), whose performance is already excellent, the spectral bound serves more as a theoretical worst-case guarantee. This point will be clearly visualized in the experimental validation in Section \ref{sec:experiments}.

\section{Engineering Implications and Design Principles}
\label{sec:engineering_implications}
The ultimate value of a profound theory lies in its ability to guide engineering practice. 
In the previous section, we established the core framework connecting network topology and service fairness (Theorem \ref{thm:spectral_bound}) and revealed its inherent bottleneck effect (Proposition \ref{prop:bottleneck_effect}). 
This section aims to distill these theoretical insights into a series of clear design principles that can be directly applied by network architects and engineers.

\subsection{Derivation and Application of the Bottleneck-Driven Design Principle}
The inequality from Theorem \ref{thm:spectral_bound}, $I \le \frac{C(1+aH)}{\ln M} \exp(-\min\{a, c\lambda_2\}H)$, is our starting point for deriving design principles. Taking the logarithm of both sides, we get:
\begin{equation}
\ln I \le \ln\left(\frac{C(1+aH)}{\ln M}\right) - \min\{a, c\lambda_2\}H,
\end{equation}
This equation indicates that above the performance threshold ($H > 0$), the logarithmic upper bound $\ln I$ is primarily dominated by the linear decay term associated with the effective threshold $H$. 
We can define a logarithmic sensitivity to measure the rate of fairness improvement. When $H$ is large, the logarithmic term of the linear pre-factor, $\ln(1+aH)$, grows slowly and can be ignored. In this case, the sensitivity is approximately:
\begin{equation}
-\frac{\partial (\ln I_{\text{upper bound}})}{\partial h_0} \approx \min\{a, c\lambda_2\}.
\end{equation}
This sensitivity quantifies the exponential improvement rate of the QoE imbalance index $I$ for each unit increase in the performance threshold $h_0$. Maximizing this improvement rate is a core goal of network optimization. Combined with the bottleneck effect from Proposition \ref{prop:bottleneck_effect}, we derive the bottleneck-driven design principle:

To maximize the improvement rate $\min\{a, c\lambda_2\}$, resource investment should be prioritized to enhance the smaller of $a$ and $c\lambda_2$. This allows us to classify the system state into two regimes:

\begin{itemize}
    \item Service-Limited State: When $a \ll c\lambda_2$, the improvement rate is approximately $a$. Here, the fairness bottleneck lies in the service parameter $a$ (e.g., the evaluation algorithm is too lenient). Investing in upgrading the network topology ($\lambda_2$) has limited effect on the improvement rate.
    
    \item Structure-Limited State: When $a \gg c\lambda_2$, the improvement rate is approximately $c\lambda_2$. Here, the fairness bottleneck lies in the connectivity of the network topology, $\lambda_2$. Further tightening the service parameter $a$ yields diminishing returns for the improvement rate.
\end{itemize}

This principle provides a clear, quantitative basis for resource allocation decisions in network optimization.

\subsection{Structure Optimization Based on Spectral Theory}
According to the analysis in Section 4.1, when the system is in a structure-limited state, the focus of optimization is to increase the network's spectral gap $\lambda_2$. This leads to a concrete question: how can we most effectively intervene in the network topology to achieve this goal? The following theorem provides an intervention strategy based on the Fiedler vector with theoretical guarantees \cite{1}.

\begin{theorem}[Fiedler Vector-Guided Structural Intervention]
\label{thm:fiedler_intervention}
Let $\mathbf v$ be the Fiedler vector of the normalized Laplacian $\mathcal L$. Consider adding an edge $(i,j) \notin E$ with a small weight $\varepsilon>0$. There exists a constant $C_{\deg}>0$ depending only on the degree ratio such that if the sufficient condition
\begin{equation}
\left(\frac{v_i}{\sqrt{d_i}}-\frac{v_j}{\sqrt{d_j}}\right)^2
\;>\;
C_{\deg}\left(\frac{v_i^2}{d_i}+\frac{v_j^2}{d_j}\right).
\end{equation}
is met, then for all sufficiently small $\varepsilon>0$, the spectral gap strictly increases, i.e., $\lambda_2(\mathcal L(\varepsilon))>\lambda_2(\mathcal L)$. Moreover, selecting the pair $(i,j)$ that maximizes the normalized Fiedler distance is the first-order optimal heuristic to satisfy this condition.
\end{theorem}

The engineering significance of this theorem is that it transforms the problem of network structure optimization from blind ``edge addition" trials into an algorithmic step with clear mathematical guidance. It provides a concrete, actionable path to resolve the structure-limited problem.

In conclusion, through the theoretical derivation and numerical analysis in this paper, the coordination between SLA and network imbalance has formed a complete set of engineering practice optimization ideas. Fig. \ref{fig:optimization} presents a closed-loop optimization process of bottlenecks-driven design principle, integrating the aforementioned theorems and corollaries into the same decision framework, and using it to guide engineers in iteratively improving the QoE fairness in actual networks.

\begin{figure}[!t]
\centering
\includegraphics[width=3.5in]{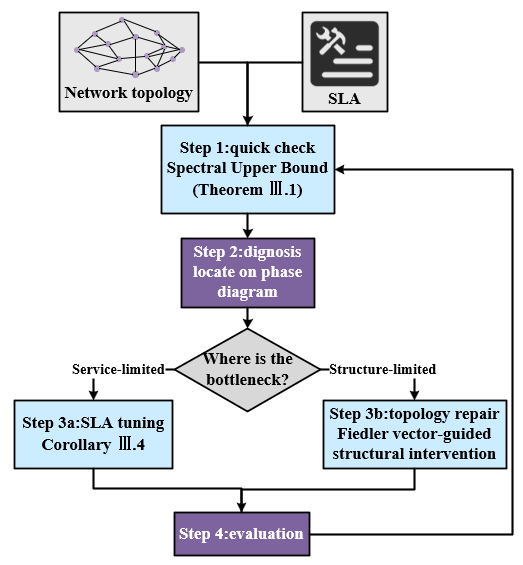}
\caption{Closed-loop optimization flowchart}
\label{fig:optimization}
\end{figure}

\subsection{Low-Complexity Calculation of Performance Upper Bounds}
The Fiedler vector-based intervention method described above relies on spectral computation, which can be costly for large-scale networks. For scenarios requiring rapid evaluation or real-time monitoring, we need a less complex performance assessment method. The following proposition provides such an alternative.

\begin{proposition}[Data-Driven Performance Certificate]
\label{prop:data_driven_cert}
For any given partition radius $0 < r \le h_0$, the average tail probability $J$ is bounded by $J \le \tau_r + e^{-a(h_0-r)}$. Consequently, using the relationship $D_{KL}(p\|u) \le J/(1-J)$, if $\tau_r + e^{-a(h_0-r)} \le 1/2$, a simpler bound holds:
\begin{equation}
I = \frac{D_{KL}(p\|u)}{\ln M} \le \frac{2}{\ln M}\Big(\tau_r(G) + e^{-a(h_0-r)}\Big).
\end{equation}
In general, the denominator in $J/(1-J)$ can be absorbed into a constant factor if $J$ is bounded away from 1.
\end{proposition}

The practical value of this proposition lies in its computational simplicity. The calculation of this upper bound does not depend on any spectral quantities. 
Its core input, $\tau_r(G)$, can be obtained through a one-time statistical analysis using standard shortest-path algorithms (such as all-pairs breadth-first search (BFS)). 
Therefore, it can serve as a low-cost check tool for performance, suitable for rapid performance audits of large-scale networks.

\section{Experimental Validation}
\label{sec:experiments}

This section aims to substantiate the theoretical framework established in Section \ref{sec:main_results} and the engineering principles derived in Section \ref{sec:engineering_implications}. 
To this end, we conduct extensive numerical experiments with two primary goals. The first goal is to verify the validity and universality of the threshold-dependent exponential spectral upper bound (Theorem \ref{thm:spectral_bound}). 
The second goal is to visualize key engineering insights to demonstrate the practical guidance value of our theoretical framework.

\subsection{Experimental Setup}
To ensure the universality and robustness of our conclusions, our experiments cover the following three representative classes of network topologies:
(i) Synthetic Models, including the Erdős-Rényi (ER), the Watts-Strogatz (WS), and the Barabási-Albert (BA) models;
(ii) a real-world topology constructed from CAIDA's AS-relationship dataset by extracting the largest connected component (LCC) of the 2025-08 snapshot (serial-1, ``20250801.as-rel.txt.bz2"), with reproducibility fingerprints reported (MD5: e1a72d19ddc87b3faba86aa8a9aebab9; 78215 nodes and 570390 edges; $\lambda_2 \approx 0.00246$); 
and (iii) an extreme-structure reference, the path graph ($P_n$), to probe bound tightness under bottleneck conditions.

On our selected topologies, we systematically scanned the SLA parameter space. In particular, the spectral gap $\lambda_2$ was computed using efficient iterative eigenvalue methods (e.g., Lanczos). The QoE imbalance index $I$ was efficiently estimated by large-scale random sampling of node pairs (its reliability is supported by Proposition \ref{prop:sampling_consistency}). For the parameters in the bound, we treated the scaling factor $c$ as a global fitting constant to enable unified cross-topology comparisons and to directly test the predicted functional form. Conversely, the threshold $r_*$ was a graph-family-dependent parameter. We computed $r_*$ directly according to the definition provided in Appendix B.

\subsection{Validation of the Spectral Upper Bound Envelope}
The core prediction of Theorem \ref{thm:spectral_bound} is that the QoE imbalance index $I$ is constrained by an exponential decay determined by $\min\{a, c\lambda_2\}H$, where $H=h_0-r_*$. 
To directly verify this prediction, we designed a semi-log envelope plot (see Fig. \ref{fig:envelope}). In the plot, the y-axis is the logarithmic imbalance $\ln(I)$, and the x-axis is the theoretically predicted ``effective work", $\min\{a, c\lambda_2\}(h_0-r_*)$.

\begin{figure}[!t]
\centering
\includegraphics[width=3.5in]{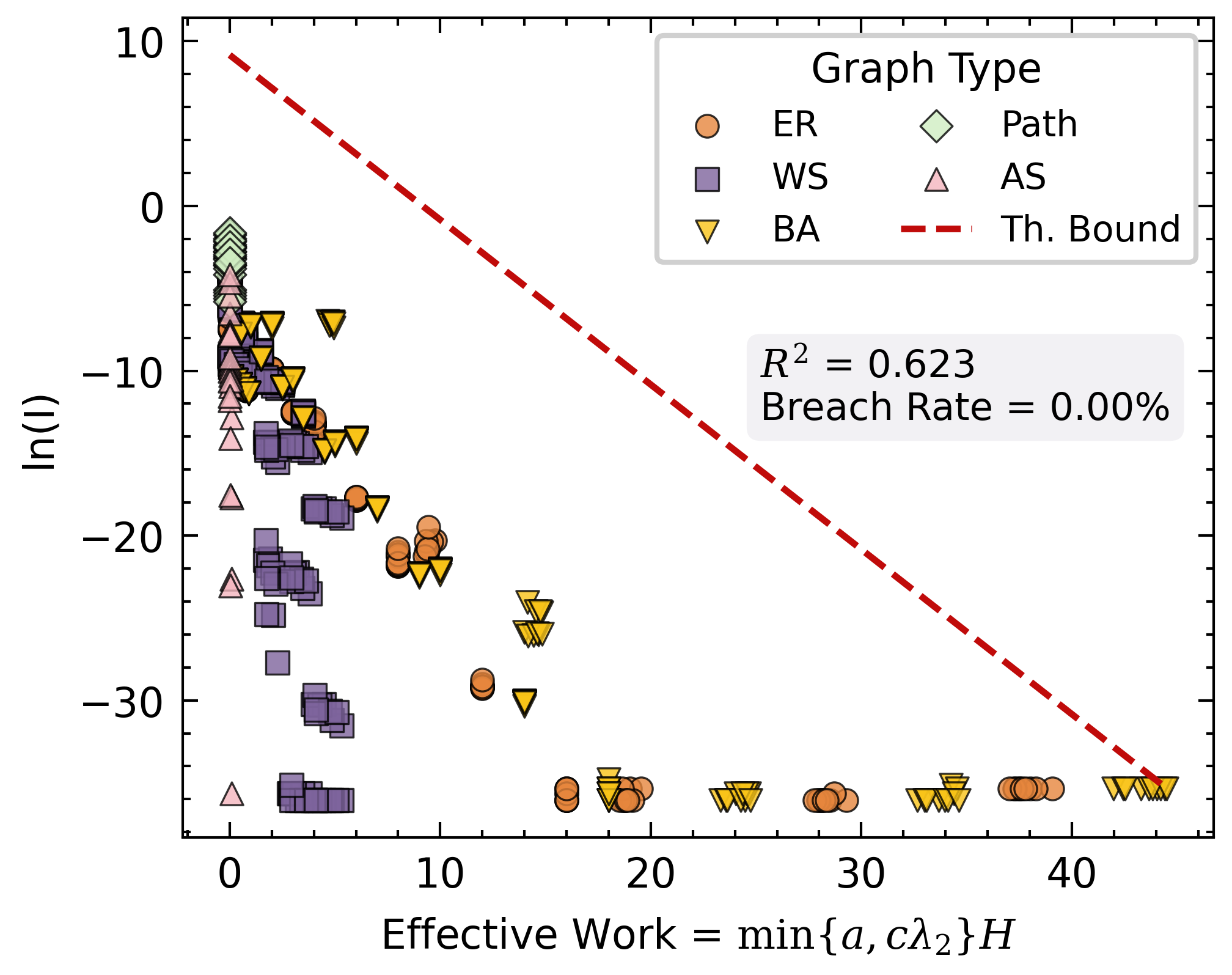}
\caption{Spectral upper bound envelope verification plot. The plot includes simulation results from all network models. Each data point represents a measurement from an independent combination (network, a, $h_0$). The red dashed line indicates the theoretical upper bound with a slope of -1. All data points are strictly located below this boundary and exhibit a significant correlation.}
\label{fig:envelope}
\end{figure}

As shown in Fig. \ref{fig:envelope}, the experimental results provide strong empirical support for Theorem \ref{thm:spectral_bound}. 
Firstly, the simulation experiments only used WS/ER/Path topologies that satisfy Assumption A to fit the global constant $c$. Then, the constant was fixed and the verification was extended on the BA/AS topologies that did not participate in the fitting. 
Across all scanned SLA parameters and topology types, we observe that $\ln(I)$ is constrained by the theoretical boundary (breach rate of 0.00\%). 
Moreover, the upper right boundary of the data cloud presents an exponential envelope shape consistent with the theory, with a slope of -1, providing direct visual evidence for the exponential upper bound form of Theorem \ref{thm:spectral_bound}. 

Importantly, the simulation results of the WS/ER/Path topologies can directly reflect the validity and correctness of Theorem \ref{thm:spectral_bound} under Assumption A. 
In contrast, BA/AS graphs typically exhibit heavy-tailed degree distributions and may violate Assumption A. Therefore, the results of the BA/AS graphs indicate that the spectral upper bound theorem still shows empirical robustness on heavy-tailed topologies. 
But due to the extremely large degree ratio, the constraint of the upper bound is relatively loose.

\subsection{Visualization of the Bottleneck Effect and the Design Phase Diagram}
To visualize the bottleneck effect, we plotted a design phase diagram of the exponential decay rate $\gamma = \min\{a, c\lambda_2\}$ in the $(a, \lambda_2)$ parameter plane.

\begin{figure}[!t]
\centering
\includegraphics[width=3.5in]{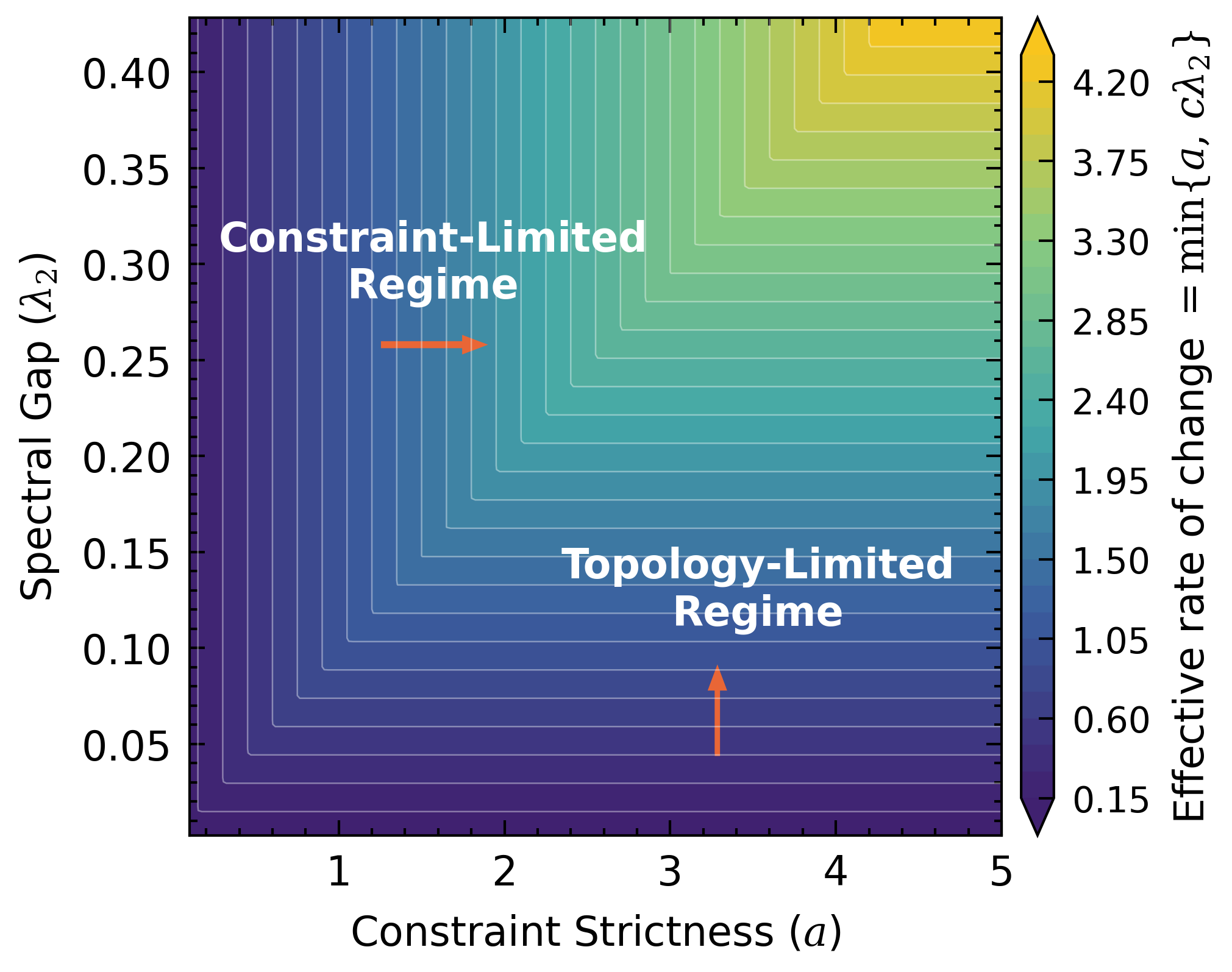}
\caption{Design phase diagram. The x-axis and y-axis represent the service evaluation stringency $a$ and network connectivity $\lambda_2$ (scaled by $c$), respectively, both on a logarithmic scale. The color represents the magnitude of the exponential decay rate $\gamma$.}
\label{fig:phase_diagram}
\end{figure}

As shown in Fig. \ref{fig:phase_diagram}, this phase diagram intuitively reveals the location of  system performance bottlenecks. 
Due to the nature of the ``min" function, the phase diagram is clearly divided by a diagonal line into two phase regions:
\begin{enumerate}
    \item Service-Limited Region (Upper-Left Area): When $a \ll c\lambda_2$, the contour lines are vertical, indicating that the exponential decay rate rate $\gamma \approx a$ is completely determined by the SLA parameter $a$.
    \item Structure-Limited Region (Lower-Right Area): When $a \gg c\lambda_2$, the contour lines are horizontal, indicating that the exponential decay rate rate $\gamma \approx c\lambda_2$, and the bottleneck lies in the network topology.
\end{enumerate}
This phase diagram provides the most direct empirical evidence for the bottleneck-driven design principle from section 4.1.

\subsection{Validation of the Data-Driven Certificate's Effectiveness}
This section aims to experimentally validate the effectiveness of Proposition \ref{prop:data_driven_cert}, as a computationally convenient performance certificate. We selected four network models: ER, BA, WS, and Path graphs. Under fixed SLA parameters, we calculated their actual QoE imbalance index $I$ and the theoretical upper bound obtained from Proposition \ref{prop:data_driven_cert}.

\begin{figure}[!t]
\centering
\includegraphics[width=3.5in]{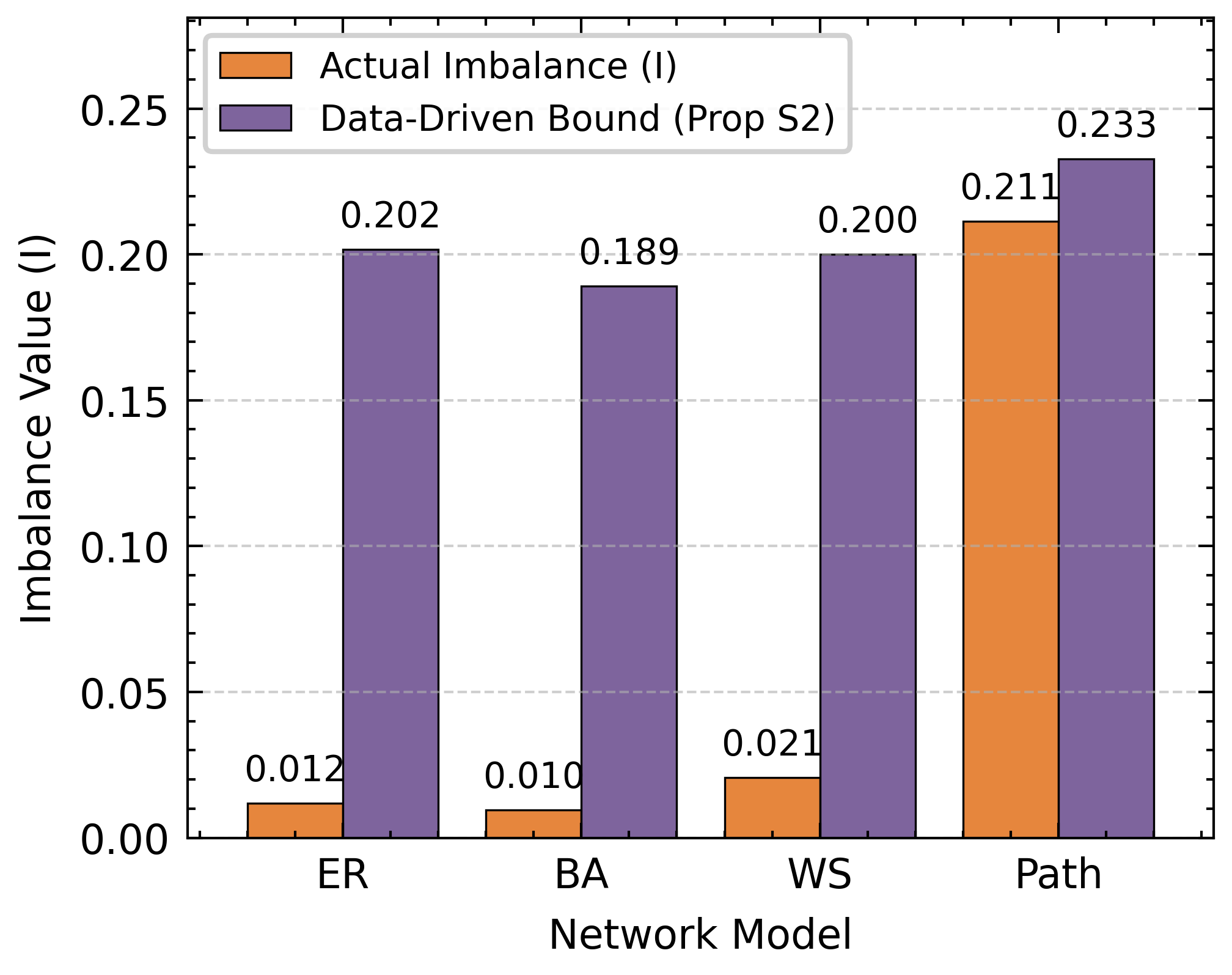}
\caption{Validation of the data-driven certificate's effectiveness. This bar chart compares the actual $I$ values with the data-driven upper bound calculated according to Proposition \ref{prop:data_driven_cert} for four network models.}
\label{fig:s2_bound}
\end{figure}

As illustrated in Fig. \ref{fig:s2_bound}, the experimental results clearly confirm the correctness of this upper bound. For all tested network types, the bars representing the ``theoretical upper bound" consistently exceed those representing the ``actual value". This evidence established Proposition \ref{prop:data_driven_cert} as a reliable and computationally low-cost performance proof.

\subsection{Case Study: Fiedler Vector-Guided Structural Intervention}
This section aims to experimentally validate the effectiveness of Theorem \ref{thm:fiedler_intervention} through a case study. We chose to experiment on a BA scale-free network. 
The first step is that we calculated the QoE imbalance index $I$ and the spectral gap $\lambda_2$ of the original network. 
Then, according to Theorem \ref{thm:fiedler_intervention}, we computed its Fiedler vector $\mathbf{v}$ and identified the non-adjacent node pair that maximizes the normalized Fiedler distance $(v_i/\sqrt{d_i} - v_j/\sqrt{d_j})^2$ to add a connection. 
Finally, we repeated this ``dianose-and-add" process several times, recording the changes in $I$ and $\lambda_2$ after each intervention.

\begin{figure}[!t]
\centering
\includegraphics[width=3.5in]{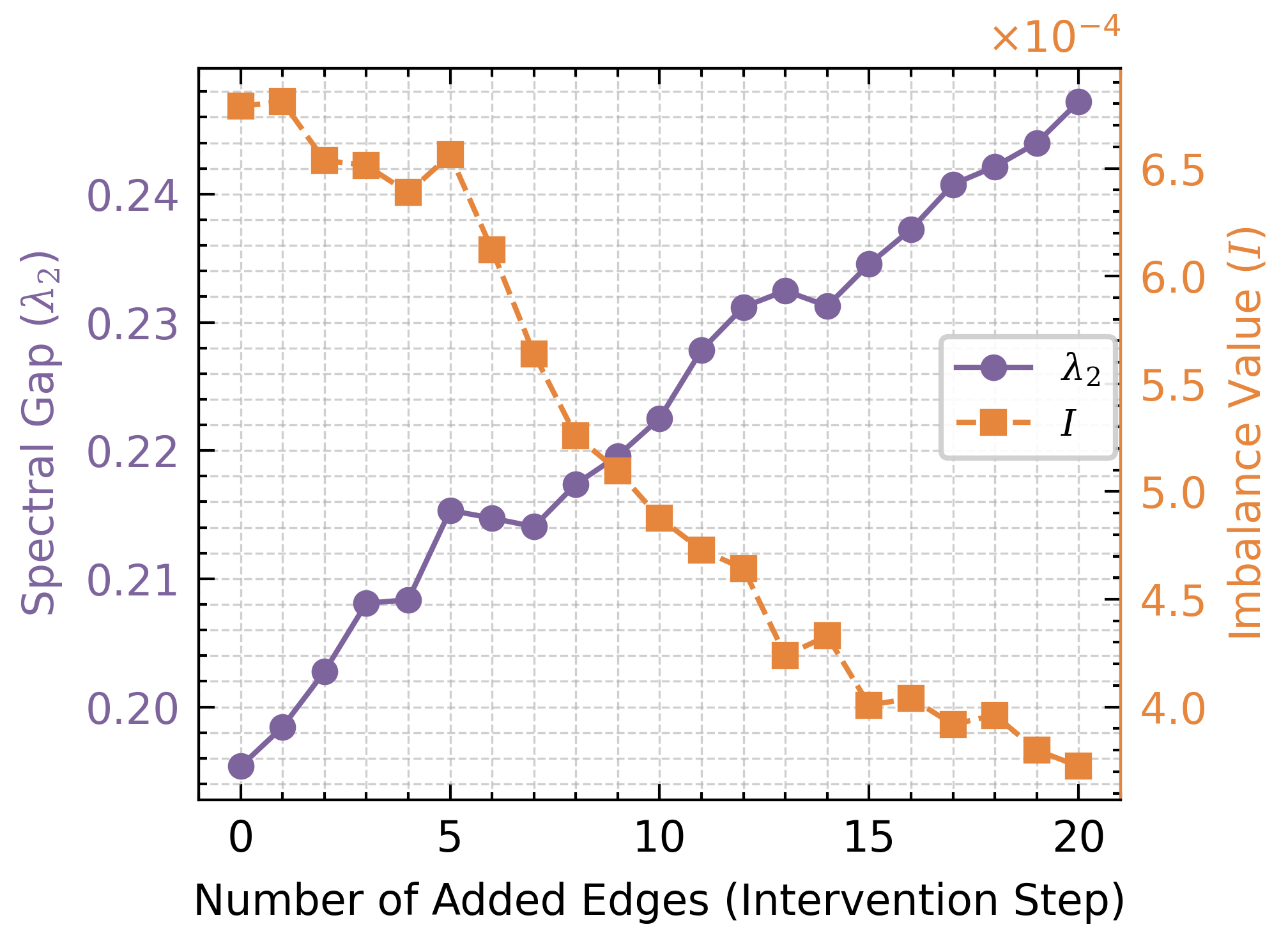}
\caption{Effectiveness of Fiedler vector-guided structural intervention. The x-axis is the number of added edges, the left y-axis is the spectral gap $\lambda_2$, and the right y-axis is the QoE imbalance index $I$.}
\label{fig:fiedler_intervention}
\end{figure}

As shown in Fig. \ref{fig:fiedler_intervention}, the experimental results provide clear evidence for the effectiveness of Theorem \ref{thm:fiedler_intervention}. 
Following each edge addition based on the normalized Fiedler distance, the network's algebraic connectivity $\lambda_2$ exhibits a significant and monotonic increasing trend. 
Correspondingly, the QoE imbalance index $I$ shows a monotonic decreasing trend. This demonstrates that topology optimization guided by spectral theory is a fundamental and reliable pathway for enhancing overall network fairness.

To evaluate the efficiency of this method in addressing structure-constrained problems, we compared the Fiedler vector-guided structural intervention strategy with three baselines: random edge addition strategy, min-degree strategy, and betweenness centrality strategy.

\begin{figure*}[!t]
  \centering
  \subfloat[Spectral gap variation]{
    \includegraphics[width=0.48\textwidth]{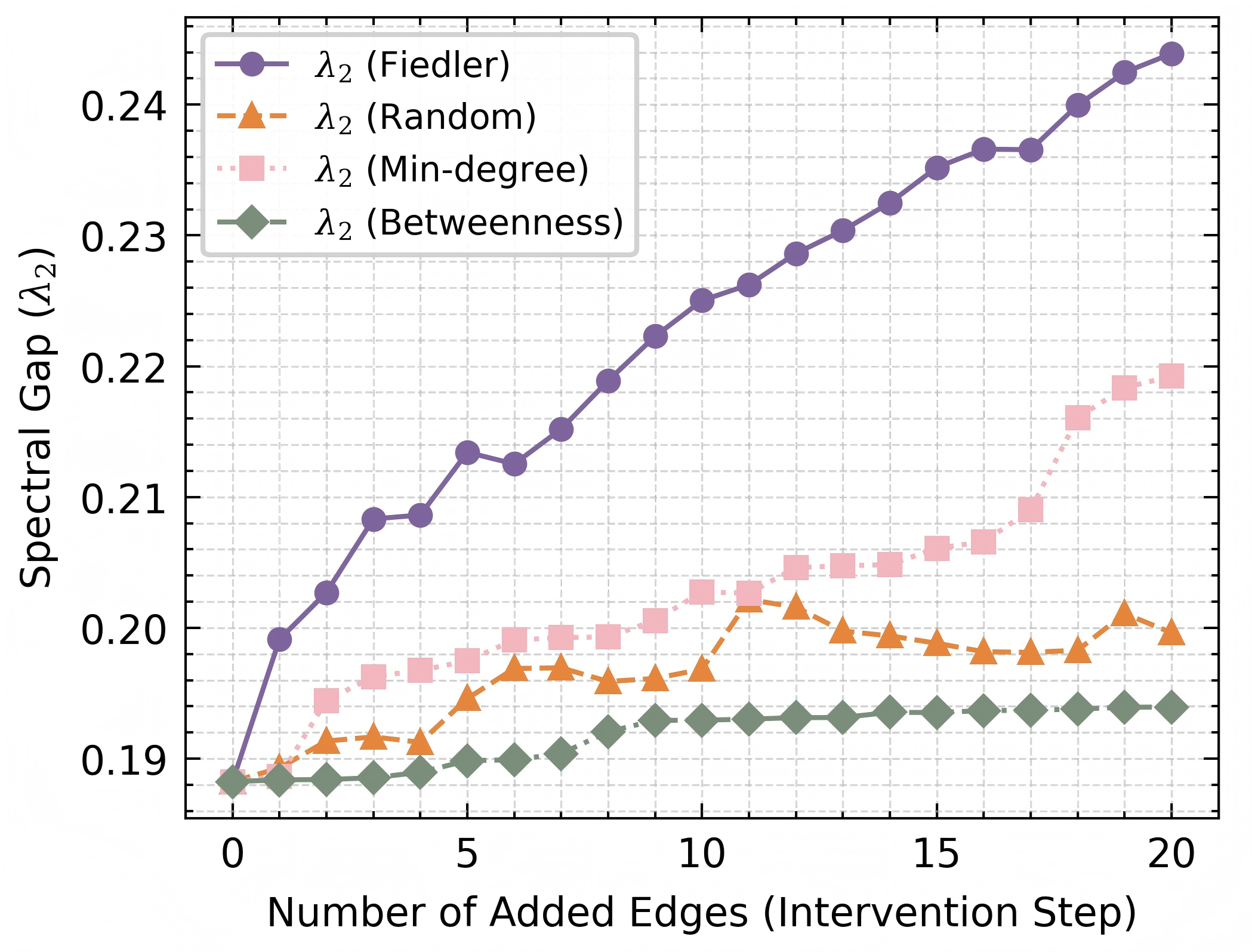}
    \label{fig:fiedler_comparison_lamad2}
  }\hfill
  \subfloat[QoE imbalance index variation]{
    \includegraphics[width=0.48\textwidth]{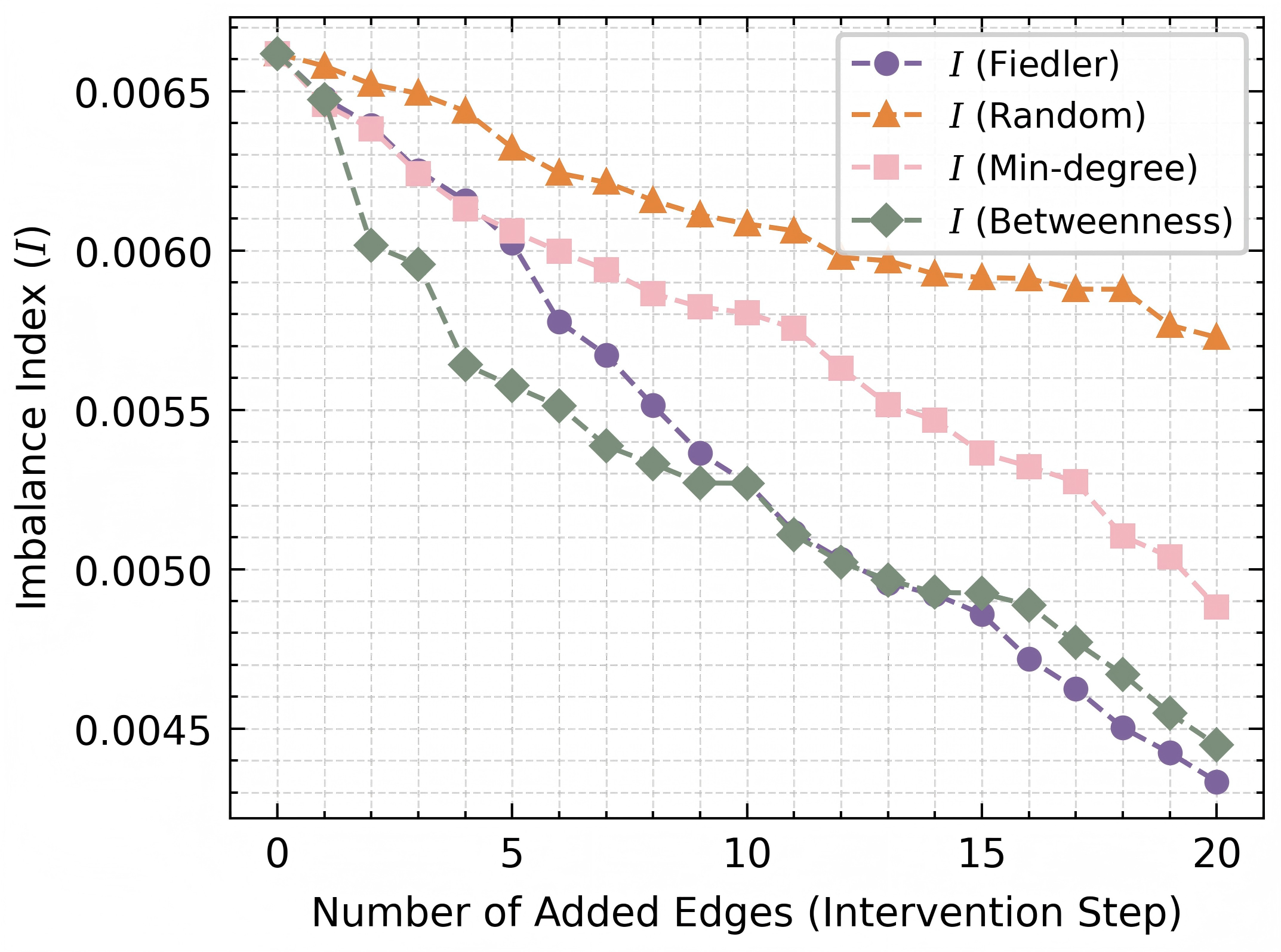}
    \label{fig:fiedler_comparison_I}
  }
  \caption{Comparison of structural intervention strategies. Subfigure (a) is the evolution of $\lambda_2$ under different structural intervention strategies. Subfigure (b) is the comparative variations of $I$ under these strategies.}
  \label{fig:fiedler_comparison}
\end{figure*}

% \begin{figure*}[!t]
% \centering
% \includegraphics[width=7in]{figure_comparison4.png}
% \caption{Comparison of structural intervention strategies. Subfigure (a) is the evolution of $\lambda_2$ under different structural intervention strategies. Subfigure (b) is the comparative variations of $I$ under these strategies.}
% \label{fig:fiedler_comparison}
% \end{figure*}

As observed in Fig. \ref{fig:fiedler_comparison}, among the four strategies, the Fiedler vector-guided structural intervention strategy yielded the most significant improvement in algebraic connectivity. 
By directly alleviating the structural bottlenecks of the topology, it strengthened global connectivity, made the path-length distribution more concentrated, and suppressed the long-tail node pairs. 
Consequently, it achieved the largest reduction in the QoE imbalance index $I$. In contrast, the betweenness centrality strategy primarily reinforced core nodes, synchronously reducing the shortest-path distances for a large number of node pairs. 
Since $I$ was highly sensitive to the overall path-length distribution, it led to a relatively rapid initial decline in $I$. However, once the paths of core nodes were sufficiently dense, the marginal benefits of this method became very limited. 
Hence, it is eventually surpassed by the Fiedler vector-guided structural intervention strategy in the later stage. The other two strategies cannot accurately address the structural bottlenecks, so the improvements in structure and the corresponding QoE imbalance index $I$ were very limited.

Overall, it can be demonstrated that the Fiedler vector-guided structural intervention strategy can efficiently and accurately mitigate severe QoE imbalance problem caused by structural bottlenecks.
\subsection{Inference Verification and Engineering Design Application}
By validating the correctness of Corollary \ref{cor:reverse_design_h0} and Corollary \ref{cor:reverse_design_lambda2}, this section provided an experimental guideline for predictable operations in network engineering design.

\begin{figure}[!t]
\centering
\includegraphics[width=3.5in]{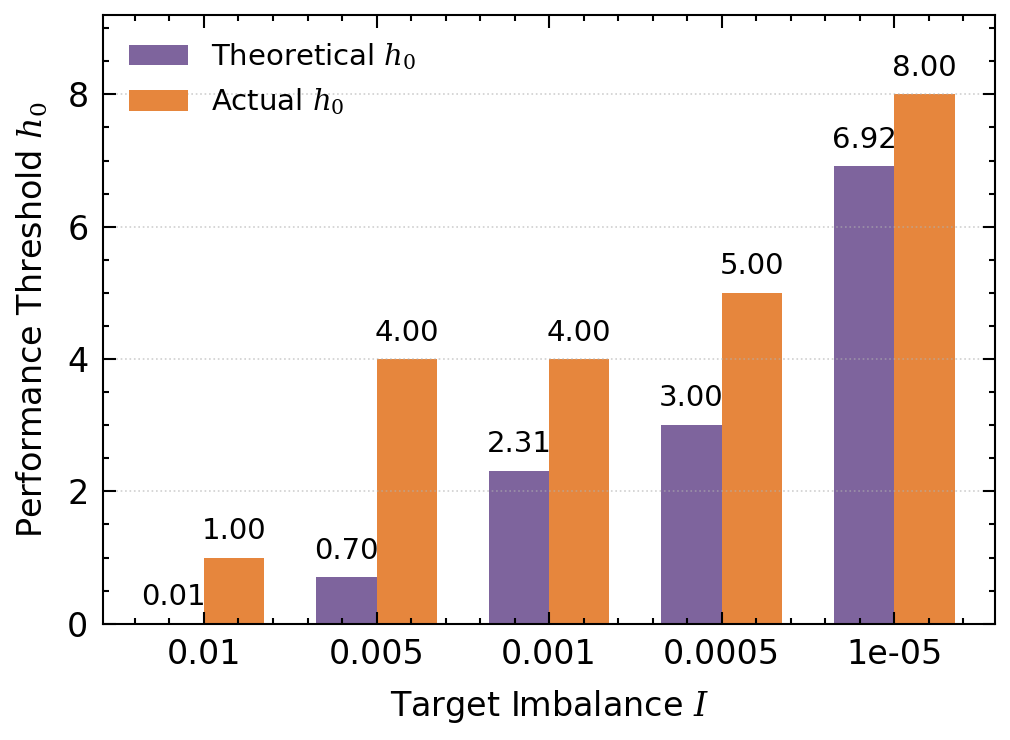}
\caption{Reverse design: threshold lower bound. This bar chart compares the actual performance threshold boundary $h_0$ with the theoretical performance threshold $h_0$ across different target $I$ values.}
\label{fig:fiedler_inference3.4}
\end{figure}

Corollary \ref{cor:reverse_design_h0} proposed a sufficient condition for the performance threshold to satisfy the target $I$. 
To verify this, we compared the relationship between the theoretical performance threshold $h_0$ and the actual boundary $h_0$ corresponding to given target $I$, as illustrated in Fig. \ref{fig:fiedler_inference3.4}.

First, the theoretical threshold $h_0$ exhibited a consistent increasing trend as the target $I$ became increasingly strict. 
This aligned perfectly with the theoretical upper bound derived in the paper. Second, according to Corollary \ref{cor:reverse_design_h0}, 
the theoretical $h_0$ serves as a guaranteed lower bound. In our simulations, for all tested target $I$, the actual $h_0$ was consistently greater than or equal to the theoretical $h_0$ boundary. 
This provides robust verification of the correctness of Corollary \ref{cor:reverse_design_h0} as a reliable performance lower bound guarantee.

Corollary \ref{cor:reverse_design_lambda2} offers theoretical guidance for predictable design directions at the level of network structural construction. 
Consequently, we conducted simulations by configuring the WS network topology to a suboptimal state. By integrating this theory with the Fiedler vector-guided structural intervention strategy, we proceeded through a workflow that evolves from ``problem diagnosis verification" to ``solution improvement." This provides a theoretical optimization example for the engineering practice of network design.

\begin{figure}[!t]
\centering
\includegraphics[width=3.5in]{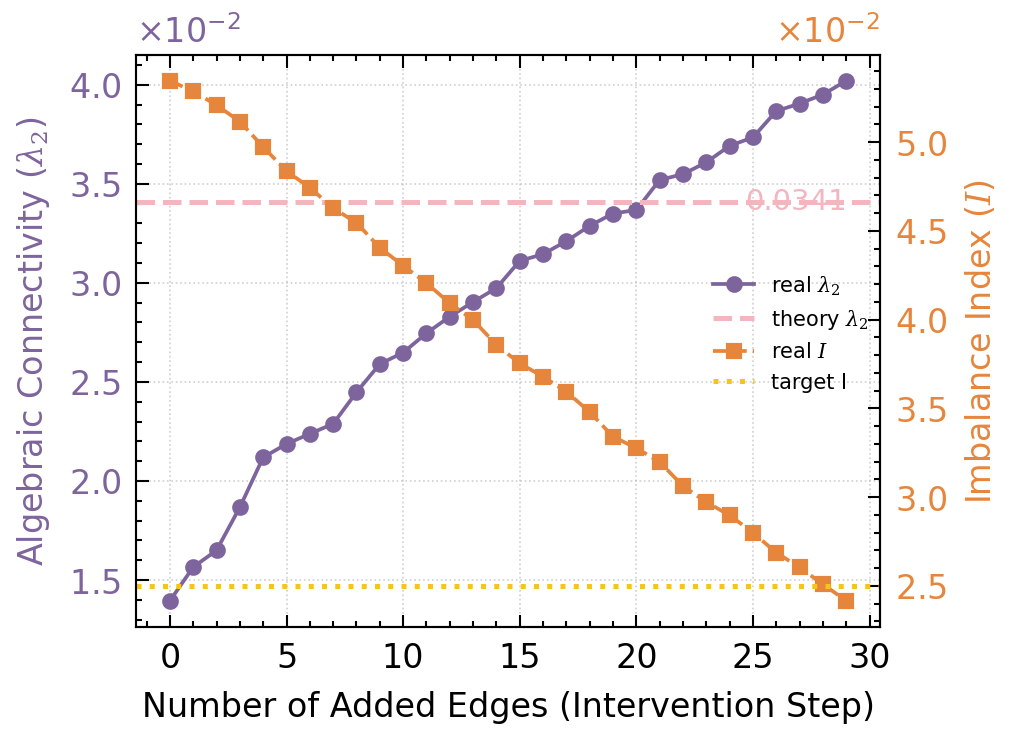}
\caption{Reverse design: spectral gap lower bound. The yellow dashed line indicates the target $I$, and the pink dashed line represents the theoretical safe connectivity required to guarantee the attainment of the target $I$.}
\label{fig:fiedler_inference3.5}
\end{figure}

As shown in Fig. \ref{fig:fiedler_inference3.5}, the simulation set a target $I$ of $0.025$. Initially, the imbalance of the network was approximately $0.043$ and the algebraic connectivity was $0.016$. By adjusting $a$ to a sufficiently large value, we ensured that the current suboptimal performance was primarily driven by the topological structure. According to Corollary \ref{cor:reverse_design_lambda2}, we calculated the theoretical connectivity threshold required to meet the target to be $0.0341$.

Consequently, we employed the Fiedler vector-guided intervention strategy to continuously improve the network's algebraic connectivity. It was observed that the algebraic connectivity exhibited an increasing trend with the addition of edges, while the network imbalance consistently decreased. This indicated that optimization of the underlying structure translated directly and effectively into improvements in upper-layer service fairness. When the network connectivity reached $0.037$, the network imbalance fell below the target $I$ for the first time. At this point, the actual required $\lambda_2$ exceeded the theoretical calculated value.

This result successfully verifies the correctness of Corollary \ref{cor:reverse_design_lambda2} as a sufficient condition providing a lower bound guarantee. By calculating the predicted connectivity threshold to assess network imbalance and performing targeted structural adjustments, this experiment clearly demonstrates that the proposed theory provides highly accurate and quantifiable engineering design goals for network optimization in ``structure-limited" states.

\section{Conclusion}
\label{sec:conclusion}
This paper moves beyond static fairness evaluation by developing a bound-based analytical framework to predict QoE fairness sensitivity under changes in SLAs. We prove a threshold-dependent exponential spectral upper bound on the QoE imbalance index $I$ (Theorem \ref{thm:spectral_bound}). 
This bound show that exponential improvement emerges only after a topology-dependent threshold ($r_*$), and that the decay rate is governed by the bottleneck term ($\min\{a, c\lambda_2\}$). From this bound, we derive a series of actionable engineering principles, including the bottleneck-driven design principle (Proposition \ref{prop:bottleneck_effect}), goal-driven reverse design methods (Corollaries \ref{cor:reverse_design_h0} and \ref{cor:reverse_design_lambda2}), and a Fiedler-vector-guided structural intervention criterion (Theorem \ref{thm:fiedler_intervention}). 

The analysis characterizes a topological ceiling under idealized assumptions, serving as a diagnostic baseline: if $\lambda_2$is insufficient even in this regime, higher-layer optimizations will be fundamentally limited. Future work will relax modeling assumptions and extend the certificate and intervention mechanisms to weighted/time-varying networks with dynamic traffic and congestion feedback.

\section*{Data Availability Statement}
We used the CAIDA AS Relationships (serial-1) dataset (2025-08 snapshot). Relevant files include ``20250801.as-rel.txt.bz2". The download link and acceptable use agreement (AUA) for the dataset are listed in the resource catalog on the official CAIDA website.

\section*{Acknowledgment}
This research adhered to CAIDA's Public Acceptable Use Agreement (Public AUA). The authors thank the CAIDA team for providing valuable public datasets.

\bibliographystyle{IEEEtran}
\bibliography{IEEEabrv, reference}

\appendices

\section{Preliminary Inequalities and Notations}
\label{app:preliminaries}

\subsection{Cheeger's Inequality and Degree Ratio Conversion}
\label{app:cheeger}
Let $G=(V,E)$ be a graph with bounded degrees: $1\le \delta\le d(v)\le \Delta$. The volume of a set $S \subset V$ is $\mathrm{vol}(S):=\sum_{v\in S} d(v)$. The volume-based Cheeger constant is $\Phi:=\min_{S\subset V,\ 0<\mathrm{vol}(S)\le \mathrm{vol}(V)/2} \frac{|\partial S|}{\mathrm{vol}(S)}$. The relationship between the spectral gap $\lambda_2$ of the normalized Laplacian and $\Phi$ is given by Cheeger's inequality:
\begin{equation}
\frac{\lambda_2}{2} \le \Phi \le \sqrt{2\lambda_2}.
\end{equation}
When degrees are bounded, $\mathrm{vol}(S)$ and $|S|$ can be mutually bounded: $\delta |S|\le \mathrm{vol}(S)\le \Delta |S|$.

\subsection{Convexity of TV and $\chi^2$}
\label{app:convexity}
For any measure space and a mixture distribution $q=\int \alpha(s)q_s ds$, the Total Variation distance $D_{TV}$ and the chi-squared divergence $\chi^2$ are convex with respect to the first argument. Therefore:
\begin{equation}
D_{TV}(q,p)\le \int \alpha(s)D_{TV}(q_s,p)ds,
\end{equation}
\begin{equation}
\chi^2(q|p)\le \int \alpha(s)\chi^2(q_s|p)ds.
\end{equation}

\subsection{Relationship between KL and $\chi^2$}
\label{app:kl_chi2}
For any probability distributions $q,p$, the Kullback-Leibler (KL) divergence and the $\chi^2$ divergence satisfy Pinsker's inequality and a first-order Taylor expansion inequality:
\begin{equation}
D_{KL}(q|p) \le \ln(1+\chi^2(q|p)) \le \chi^2(q|p).
\end{equation}

\subsection{Pointwise Bound for Logistic Weights}
\label{app:logistic_bound}
Let the weight function be $s(r)=1/(1+e^{a(r-h_0)})$. Its derivative satisfies:
\begin{equation}
\frac{ds}{dr} = -a \cdot s(r)(1-s(r)),
\end{equation}
And it has the following symmetric exponential bound:
\begin{equation}
s(r)(1-s(r))=\frac{1}{2+2\cosh(a(r-h_0))} \le e^{-a|r-h_0|}.
\end{equation}
This inequality is derived from $\cosh x \ge e^{|x|}/2$.

\subsection{Pair-Ball Counting Identity}
\label{app:pair_ball_identity}
Let $B_r(u)$ be the ball of radius $r$ centered at $u$ (in graph distance). Then the fraction of unordered pairs $\tau_r$ with distance greater than $r$ satisfies:
\begin{equation}
\tau_r = \frac{1}{M} \cdot \frac{1}{2}\sum_{u\in V} |V\setminus B_r(u)|.
\end{equation}
Reason: Each unordered pair $\{i,j\}$ with $\mathrm{dist}(i,j)>r$ is counted exactly twice in the summation on the right: once when $u=i$ (counting $j \in V\setminus B_r(i)$) and once when $u=j$ (counting $i \in V\setminus B_r(j)$).

\textbf{Remark:} Assumption A is introduced to keep the constants in our spectral bounds independent of the network size, which is standard in spectral-inequality analyses. 
When the degree ratio $\beta = \Delta/\delta$ grows with $n$ (e.g., in heavy-tailed graphs), the same derivations still hold but the bound constants become $\beta$-dependent and the resulting bound may be looser. 
In this work, we use Assumption A to establish the cleanest form of the bound.

\section{Proof of the Spectral Tail Bound Lemma}
\label{app:proof_spectral_tail}

\begin{lemma}[Spectral Tail Bound]
\label{lem:spectral_tail_bound}
This is the key lemma required for Theorem \ref{thm:spectral_bound} in the main text. There exist constants $C,c>0$ (depending only on $\Delta/\delta$) and a threshold $r_* \le C \frac{\ln n}{\lambda_2}$, such that for all $r\ge 0$:
\begin{equation}
\tau_r \le C \exp\big(-c\lambda_2 [r-r_*]_+\big),
\end{equation}
\end{lemma}
\begin{proof}
Fix a source node $u$. Let $B_t:=B_t(u)$. When $\mathrm{vol}(B_t)\le \mathrm{vol}(V)/2$, by Cheeger's inequality (Appendix \ref{app:cheeger}), we have $|\partial B_t|\ge \Phi \mathrm{vol}(B_t)$. The volume growth satisfies:
\begin{equation}
\begin{split}
    \mathrm{vol}(B_{t+1})-\mathrm{vol}(B_t) \ge \delta |B_{t+1}\setminus B_t| \\
    \ge \frac{\delta}{\Delta} |\partial B_t| \ge \frac{\delta \Phi}{\Delta}\mathrm{vol}(B_t),
\end{split}
\end{equation}
Therefore, $\mathrm{vol}(B_{t+1})\ge (1+\kappa)\mathrm{vol}(B_t)$, where $\kappa:=\frac{\delta\Phi}{\Delta} \ge c\lambda_2$. The number of steps $t_*(u)$ required to grow from $\mathrm{vol}(B_0)\ge \delta$ to $\mathrm{vol}(V)/2$ satisfies:
\begin{equation}
t_*(u) \le \frac{\log(\mathrm{vol}(V)/(2\delta))}{\log(1+\kappa)} \le C\frac{\log n}{\lambda_2},
\end{equation}
Let $r_* = \sup_u t_*(u)$. When $t > r_*$, we have $\mathrm{vol}(B_t) > \mathrm{vol}(V)/2$. Applying Cheeger's inequality to the complement set $S_t=V\setminus B_t$, we get exponential shrinkage of the complement's volume: $\mathrm{vol}(S_{t+1})\le (1-\kappa)\mathrm{vol}(S_t)$. By recursion:
\begin{equation}
\begin{split}
    |V\setminus B_t(u)| \le \frac{1}{\delta}\mathrm{vol}(S_t) \le \frac{\mathrm{vol}(V)}{2\delta}(1-\kappa)^{[t-t_*(u)]_+}\\  
\le \frac{n\Delta}{2\delta} \exp(-c\lambda_2 [t-t_*(u)]_+),
\end{split}
\end{equation}
Averaging over all source nodes $u$ and using the identity from Appendix \ref{app:pair_ball_identity}:
\begin{equation}
\begin{split}
    \tau_r = \frac{1}{2M}\sum_{u}|V\setminus B_r(u)| \le \frac{1}{2M}\sum_u \frac{n\Delta}{2\delta} e^{-c\lambda_2[r-t_*(u)]_+}\\
\le C_0 e^{-c\lambda_2[r-r_*]_+}.
\end{split}
\end{equation}
where the constant $C_0$ absorbs factors like $\frac{n^2 \Delta}{2\delta M} \approx \frac{\Delta}{\delta}$.
\end{proof}

\section{Proofs of Lemmas related to Tomographic Decomposition}
\label{app:proof_tomographic}

\begin{lemma}[Tomographic Representation]
\label{lem:tomographic_rep}
$w(h)=\int_0^1 \mathbf{1}_{\{h \le r(s)\}} ds$, where $r(s)=h_0+\frac{1}{a}\log\frac{1-s}{s}$.
\end{lemma}
\begin{proof}
The function $w(h)$ is monotonically decreasing and satisfies $w(h)>s \iff h<r(s)$. Thus:
\begin{equation}
\begin{split}
    w(h)=\int_0^1 \mathbf{1}_{\{s<w(h)\}}ds = \int_0^1 \mathbf{1}_{\{h<r(s)\}}ds\\
    = \int_0^1 \mathbf{1}_{\{h\le r(s)\}}ds.
\end{split}
\end{equation}
\end{proof}

\begin{lemma}[Mixture Representation and Convexity]
\label{lem:mixture_rep}
(a) $p=\int_0^1 \alpha(s) u_{\le r(s)} ds$, where $u_{\le r}$ is the uniform distribution on the set $S_r=\{k:h_k\le r\}$, and $\alpha(s)=|S_{r(s)}|/W$ is a probability density.
(b) $D_{TV}(p,u)\le \int_0^1 \alpha(s) D_{TV}(u_{\le r(s)},u) ds$.
\end{lemma}
\begin{proof}
(a) From Lemma \ref{lem:tomographic_rep} and Fubini's theorem:
$W=\sum_k w(h_k)=\sum_k\int_0^1 \mathbf{1}_{\{h_k\le r(s)\}}ds = \int_0^1 |S_{r(s)}| ds$.
Thus, $p_k=\frac{w(h_k)}{W}=\int_0^1 \frac{|S_{r(s)}|}{W}\cdot \frac{\mathbf{1}_{\{h_k\le r(s)\}}}{|S_{r(s)}|}ds = \int_0^1 \alpha(s)u_{\le r(s)}(k)ds$.
(b) This follows from the convexity of $D_{TV}$ with respect to its first argument (Appendix \ref{app:convexity}).
\end{proof}

\begin{lemma}[TV of a Truncated Model]
\label{lem:tv_truncated}
$D_{TV}(u_{\le r},u)=\tau_r$.
\end{lemma}
\begin{proof}
On the set $S_r$, the difference is $|u_{\le r}-u|=|1/|S_r|-1/M|$; on the complement $T_r$, the difference is $1/M$. The total variation distance is half the L1-norm:
\begin{equation}
\begin{split}
    D_{TV}(u_{\le r},u)=\frac{1}{2}\left(|S_r|\left|\frac{1}{|S_r|}-\frac{1}{M}\right|+|T_r|\cdot\frac{1}{M}\right)\\
= \frac{M-|S_r|}{M}=\frac{|T_r|}{M}=\tau_r.
\end{split}
\end{equation}
\end{proof}

\section{Proofs of Relationships between J and TV / $\chi^2$ / KL}
\label{app:proof_J_to_div}

\begin{definition}[Average Tail Probability J]
$J:=\int_0^1 \tau_{r(s)} ds = \mathbb{E}_{\mathbf{u}}[1-w(h)]$.
\end{definition}

\begin{proposition}[From J to TV]
\label{prop:J_to_TV}
$D_{TV}(p,u) \le \frac{J}{1-J}$.
\end{proposition}
\begin{proof}
From the lemmas in Appendix \ref{app:proof_tomographic}, $D_{TV}(p,u)\le \int \alpha(s)\tau_{r(s)}ds$.
Also, $\frac{W}{M}=\int_0^1 \frac{|S_{r(s)}|}{M}ds=\int_0^1 (1-\tau_{r(s)})ds=1-J$.
Substituting $\alpha(s) = \frac{M(1-\tau_{r(s)})}{W} = \frac{1-\tau_{r(s)}}{1-J}$:
\begin{equation}
\begin{split}
    D_{TV}(p,u) \le \frac{1}{1-J}\int (1-\tau_{r(s)})\tau_{r(s)} ds\\
= \frac{J-\int \tau_{r(s)}^2 ds}{1-J}\le \frac{J}{1-J}.
\end{split}
\end{equation}
\end{proof}

\begin{proposition}[From J to $\chi^2$ and KL]
\label{prop:J_to_KL}
$D_{KL}(p|u) \le \chi^2(p|u) \le \frac{J}{1-J}$.
\end{proposition}
\begin{proof}
By the convexity of $\chi^2(\cdot|u)$ (Appendix \ref{app:convexity}), $ \chi^2(p|u)\le \int \alpha(s)\chi^2(u_{\le r(s)}|u)ds$.
Direct calculation yields $\chi^2(u_{\le r}|u)=\sum_k \frac{u_{\le r}(k)^2}{u(k)}-1 = \sum_{k\in S_r} \frac{(1/|S_r|)^2}{1/M}-1=\frac{M}{|S_r|}-1 = \frac{\tau_r}{1-\tau_r}$.
Substituting $\alpha(s)=\frac{1-\tau_{r(s)}}{1-J}$, we get $\chi^2(p|u)\le \frac{1}{1-J}\int (1-\tau_{r(s)})\frac{\tau_{r(s)}}{1-\tau_{r(s)}} ds = \frac{J}{1-J}$.
The relationship $D_{KL}\le \chi^2$ follows from Appendix \ref{app:kl_chi2}.
\end{proof}

\section{Proof of the Exponential Upper Bound on J}
\label{app:proof_J_bound}

\begin{theorem}[Exponential Upper Bound on J]
\label{thm:J_bound}
Let $H:=h_0-r_*\ge 0$. There exist constants $C,c>0$ (depending only on $\Delta/\delta$) such that:
\begin{equation}
J \le C\left( e^{-aH} + aH e^{-\min(a,c\lambda_2)H} + \frac{a}{a+c\lambda_2}e^{-c\lambda_2 H}\right).
\end{equation}
\end{theorem}
\begin{proof}
By definition and Appendix \ref{app:logistic_bound}, $J=\int_{-\infty}^{+\infty} \tau_r a s(r)(1-s(r))dr$. Substituting the spectral tail bound from Lemma \ref{lem:spectral_tail_bound} (Appendix \ref{app:proof_spectral_tail}) and the pointwise bound from Appendix \ref{app:logistic_bound}:
\begin{equation}
J \le C\int_{-\infty}^{+\infty} e^{-c\lambda_2[r-r_*]_+} a e^{-a|r-h_0|} dr,
\end{equation}
We split the integral into three parts: A: $(-\infty, r_*]$, B: $[r_*, h_0]$, and C: $[h_0, \infty)$.
\textbf{Part A ($r \le r_*$):} $[r-r_*]_+=0$, $\tau_r \le C$.
\begin{equation}
I_A \le C \int_{-\infty}^{r_*} a e^{-a(h_0-r)}dr = C e^{-a(h_0-r_*)} = C e^{-aH},
\end{equation}
\textbf{Part B ($r_* \le r \le h_0$):} Let $t=r-r_* \in [0, H]$.
\begin{equation}
\begin{split}
    I_B \le C a \int_{r_*}^{h_0} e^{-c\lambda_2(r-r_*)} e^{-a(h_0-r)}dr\\
= C a e^{-aH}\int_0^H e^{-(c\lambda_2-a)t}dt,
\end{split}
\end{equation}
The result of this integral is $C a e^{-aH}\frac{1-e^{-(c\lambda_2-a)H}}{c\lambda_2-a}$ (when $c\lambda_2 \neq a$), or $C a H e^{-aH}$ (when $c\lambda_2 = a$). A unified upper bound for both cases is $C aH e^{-\min(a,c\lambda_2)H}$.
\textbf{Part C ($r \ge h_0$):} Let $t=r-h_0 \ge 0$.
\begin{equation}
\begin{split}
    I_C \le C a \int_{h_0}^{\infty} e^{-c\lambda_2(r-r_*)} e^{-a(r-h_0)}dr\\
= C a e^{-c\lambda_2H}\int_0^\infty e^{-(a+c\lambda_2)t}dt = C \frac{a}{a+c\lambda_2} e^{-c\lambda_2H}
\end{split}
\end{equation}
Summing the three parts yields the theorem's conclusion.
\end{proof}

\section{Proof of Theorem 1 (Exponential Upper Bound on I)}
\label{app:proof_theorem1}
\begin{proof}[Proof of Theorem \ref{thm:spectral_bound}]
From Proposition \ref{prop:J_to_KL} in Appendix \ref{app:proof_J_to_div}, we have $D_{KL}(p|u) \le \frac{J}{1-J}$.
From Theorem \ref{thm:J_bound} in Appendix \ref{app:proof_J_bound}, we have $J \le C'(1+aH)e^{-\min(a,c\lambda_2)H}$.
When $H$ is large enough such that $J \le 1/2$, we have $\frac{J}{1-J}\le 2J$. In general, the constant 2 can be absorbed into $C$. Thus:
\begin{equation}
D_{KL}(p|u) \le C(1+aH)e^{-\min(a,c\lambda_2)H}.
\end{equation}
Dividing both sides by $\ln M$ yields the theorem.
\end{proof}

\section{Proof of Theorem 2 (Structural Intervention)}
\label{app:proof_theorem2}
\begin{proof}[Proof of Theorem \ref{thm:fiedler_intervention}]
Let $\mathcal{L}(\varepsilon)=I-D(\varepsilon)^{-1/2}W(\varepsilon)D(\varepsilon)^{-1/2}$. We compute the derivative of $\lambda_2(\varepsilon)$ at $\varepsilon=0$. By the Hellmann–Feynman theorem, $\frac{d\lambda_2}{d\varepsilon}|_0 = v^T \dot{\mathcal{L}} v$.
Let $\mathcal{S}=D^{-1/2}WD^{-1/2}$, so $\dot{\mathcal{L}} = -\dot{\mathcal{S}}$.
\begin{equation}
    \dot{\mathcal{S}} = (\dot{D}^{-1/2})WD^{-1/2} + D^{-1/2}\dot{W}D^{-1/2} + D^{-1/2}W(\dot{D}^{-1/2})
\end{equation}
Here $\dot{W}=E_{ij}$ (edge matrix), $\dot{D}=E_{ii}+E_{jj}$ (degree matrix), and $\frac{d}{d\varepsilon}D^{-1/2}|_0 = -\frac{1}{2}D^{-3/2}\dot{D}$.
Substituting and computing $v^T \dot{\mathcal{L}} v = -v^T \dot{\mathcal{S}} v$. Let $x_k = v_k/\sqrt{d_k}$.
The middle term contributes $-v^T D^{-1/2} E_{ij} D^{-1/2} v = -(x_i-x_j)^2 + (x_i^2+x_j^2)$. The degree perturbation terms on both sides contribute $\sum_{u \in \{i,j\}} x_u \frac{(Wx)_u}{d_u}$.
Combining all terms, after a series of inequality relaxations, it can be shown that there exists a constant $C_{\deg}$ depending only on the degree ratio such that:
\begin{equation}
\frac{d\lambda_2}{d\varepsilon}|_0 \ge \left(\frac{v_i}{\sqrt{d_i}}-\frac{v_j}{\sqrt{d_j}}\right)^2 - C_{\deg}\left(\frac{v_i^2}{d_i}+\frac{v_j^2}{d_j}\right).
\end{equation}
To ensure the derivative is positive, one should maximize the first term. Therefore, choosing the node pair that maximizes the normalized Fiedler distance is the first-order optimal strategy to guarantee an increase in $\lambda_2$.
\end{proof}

\section{Proof of Proposition \ref{prop:data_driven_cert} (Data-Driven Certificate)}

\label{app:proof_prop3}
\begin{proof}[Proof of Proposition \ref{prop:data_driven_cert}]
We start from $J=\int_{-\infty}^{+\infty} \tau_\rho a s(\rho)(1-s(\rho))d\rho$. We split the integral at point $r$.
\textbf{Interval $(\boldsymbol{-\infty, r]}$:} Here, $\tau_\rho \le 1$.
\begin{equation}
\int_{-\infty}^{r} \tau_\rho a s(1-s) d\rho \le \int_{-\infty}^{r} a e^{-a(h_0-\rho)} d\rho = e^{-a(h_0-r)},
\end{equation}
\textbf{Interval $[\boldsymbol{r, \infty)}$:} Here, $\tau_\rho \le \tau_r$ (because $\tau$ is non-increasing).
\begin{equation}
\int_{r}^{+\infty} \tau_\rho a s(1-s) d\rho \le \tau_r \int_{r}^{+\infty} a s(1-s) d\rho \le \tau_r.
\end{equation}
Adding the two parts gives $J \le \tau_r + e^{-a(h_0-r)}$. The result follows from $I \le C \cdot J/\ln M$.
\end{proof}

\section{Proof of Proposition \ref{prop:sampling_consistency} (Sampling Consistency)}
\label{app:proof_prop4}
\begin{proposition}[Sampling Consistency]
\label{prop:sampling_consistency}
$|\widehat I-I|=O_{\mathbb P}\Big(\sqrt{\frac{\log M}{m}}\Big)$, where $\widehat{I}$ is the imbalance index estimated from $m$ samples.
\end{proposition}
\begin{proof}
Let $\widetilde p$ be the empirical frequency without smoothing. By concentration inequalities for multinomial distributions (e.g., Bernstein's inequality) and applying a union bound over all $M$ components, we know that with high probability $1-M^{-c}$:
\begin{equation}
\|\widetilde p-p\|_{\infty} \le C\sqrt{\frac{\log M}{m}},
\end{equation}
Smoothing $\widetilde p$ with a Dirichlet-$\alpha$ prior gives $\widehat p$. When $\alpha M/m \to 0$, $\widehat p$ maintains the same order of convergence bound in the infinity norm.
The imbalance index $I$ is a function of the distribution $p$. When the minimum component of $p$ is positively lower-bounded (guaranteed by smoothing), $I(p)$ is Lipschitz continuous in a neighborhood of $p$. More precisely, the KL divergence can be expanded via a second-order Taylor series:
\begin{equation}
D_{KL}(\widehat p|u)-D_{KL}(p|u) = \langle \nabla D_{KL}(p|u), \widehat p-p \rangle + O(\|\widehat p-p\|_2^2).
\end{equation}
Using the norm relationships $\|\cdot\|_1 \le \sqrt{M}\|\cdot\|_2 \le M\|\cdot\|_{\infty}$, the infinity norm bound can be converted to bounds on the L1 and L2 norms. Ultimately, we find that $|D_{KL}(\widehat p|u)-D_{KL}(p|u)|$ is bounded by $O(\sqrt{\frac{\log M}{m}})$ with high probability. Dividing both sides by $\ln M$ yields the conclusion.
\end{proof}

\end{document}